\documentclass[a4paper,onecolumn,11pt,accepted=2024-11-20]{quantumarticle}
\pdfoutput=1
\usepackage[utf8]{inputenc}
\usepackage[english]{babel}
\usepackage[T1]{fontenc}
\usepackage{hyperref}
\usepackage[
backend=biber,
style=alphabetic,
sorting=ynt,
maxbibnames=99
]{biblatex}
\usepackage{csquotes}

\usepackage{tikz}
\usepackage{lipsum}

\usepackage{physics} 
\usepackage{enumitem} 
\usepackage{graphicx} 
\usepackage{comment} 
\usepackage{multirow} 
\usepackage{float} 
\usepackage[noabbrev,nameinlink]{cleveref} 
\usepackage{caption, subcaption} 
\usepackage{amssymb} 
\usepackage{appendix}
\usepackage{amsthm} 
\usepackage{mdframed} 
\usepackage{authblk} 
\usepackage{mathrsfs} 
\usepackage{array} 
\usepackage{longtable} 
\usepackage{mathtools} 
\usepackage{algorithm} 
\usepackage{algorithmic} 
\usepackage{setspace} 

\newtheorem{theorem}{Theorem}[section]

\newtheorem{definition}{Definition}[section]
\newtheorem{remark}{Remark}[section]

\crefname{theorem}{Thm.}{Thms.}
\Crefname{theorem}{Theorem}{Theorems}
\crefname{equation}{Eq.}{Eqs.}
\Crefname{equation}{Equation}{Equations}

\newcommand{\C}{\mathbb{C}}
\newcommand{\eps}{\varepsilon}

\newcommand{\Linear}{\mathcal{L}}
\newcommand{\U}{\mathcal{U}}
\renewcommand{\P}{\mathcal{P}}

\renewcommand{\H}{\mathcal{H}}

\newcommand{\eqdef}{\vcentcolon=}
\newcommand{\sbline}{\\[.3\normalbaselineskip]} 
\crefname{algorithm}{Protocol}{Protocols} 

\begin{document}

\title{Port-Based State Preparation and Applications}

\author[1]{Garazi Muguruza}
\email{g.muguruzalasa@uva.nl}
\orcid{0009-0006-4627-9520}
\author[1]{Florian Speelman}
\email{f.speelman@uva.nl}
\orcid{0000-0003-3792-9908}
\affiliation[1]{QuSoft \& Informatics Institute, University of Amsterdam,
Netherlands}

\maketitle

\begin{abstract}
    We introduce Port-Based State Preparation (PBSP), a teleportation task where Alice holds a complete classical description of the target state and Bob's correction operations are restricted to only tracing out registers. We show a protocol that implements PBSP with error decreasing exponentially in the number of ports, in contrast to the polynomial trade-off for the related task of Port-Based Teleportation, and we prove that this is optimal when a maximally entangled resource state is used.

    As an application, we introduce approximate Universal Programmable Hybrid Processors (UPHP). Here the goal is to encode a unitary as a quantum state, and the UPHP can apply this unitary to a quantum state when knowing its classical description. 
    We give a construction that needs strictly less memory in terms of dimension than the optimal approximate Universal Programmable Quantum Processor achieving the same error. Additionally, we provide lower bounds for the optimal trade-off between memory and error of UPHPs.
\end{abstract}

\section{Introduction}
Quantum teleportation~\cite{bennett_teleporting_1993} is a fundamental task in quantum information processing: Two parties, Alice and Bob, can implement a quantum channel by means of a classical channel and one entangled quantum bit shared between the two parties. Although an infinitely-precise description of a quantum state would require an infinite number of classical bits to describe, quantum teleportation allows an observer who cannot access the infinite information in a quantum state to send this information only with 2-bits of classical communication and an entangled qubit.

Standard teleportation is a protocol where the transmitted quantum state is unknown, but what if Alice knows which quantum state she would like to transmit? Of course, she could first create the state and then use a standard teleportation protocol, but it might open up the option for other protocols. This variant of teleportation is called remote state preparation (RSP) \cite{bennett_remote_2001}, and can indeed be performed more efficiently; for instance, asymptotically only one bit of classical communication is necessary to transmit a qubit.

Port-based teleportation~\cite{ishizaka_asymptotic_2008} (PBT) is a variant of teleportation where the receiver's correction operation is very simple; after receiving classical information from Alice, Bob simply picks one of the subsystems, or \emph{ports}, and traces out the rest of his part of the entangled resource state. This form of the correction operation allows Bob to perform quantum operations on the output of the protocol before receiving the correct port -- this property makes PBT a useful primitive for various tasks in quantum information, including proving bounds on channel discrimination~\cite{pirandola_fundamental_2019}, universal programmable quantum processors \cite{nielsen_programmable_1997,yang_optimal_2020}, constructing protocols for instantaneous non-local quantum computation~\cite{beigi_simplified_2011}, transposing and inverting unitary operations~\cite{quintino_probabilistic_2019,quintino_reversing_2019}, storage and retrieval of unitary channels~\cite{sedlak_optimal_2019} and communication complexity and Bell nonlocality~\cite{buhrman_quantum_2016}. 

The primitive PBT is powerful but not very efficient, requiring a number of ports scaling polynomially with the dimension of the transmitted state (i.e., exponential in the number of qubits transmitted) and the tolerated error. Given that knowing the teleported quantum state makes RSP more efficient than standard teleportation, we will study the following question:\\
\vspace{0.1cm}\\
\makebox[\textwidth]{\emph{What is the complexity of port-based teleportation of a known state?}}\\
\makebox[\textwidth]{\emph{And what would be the uses of such primitive?}}\\
\vspace{0.1cm}

\paragraph{Our contributions.}
In this work, we introduce and study a variant of quantum teleportation, Port-Based State Preparation (PBSP); where the target qubit is known to Alice, but is still unknown to Bob, and Bob cannot perform any correction operation other than discarding ports. We can consider probabilistic and deterministic variants of PBSP, where the former either constructs the desired state perfectly or signals a failure, and the latter always constructs a (necessarily approximate) state. A protocol for this task is present in the RSP literature, denoted as the `column method'~\cite{bennett_remote_2005}, but as far as we are aware we are the first ones describing the general task explicitly and studying the relations between the dimension of the entangled resource state, the number of ports and the error of the outcome state.

Since Alice knows the state she intends to send, her measurement can depend on it, which turns out to be enough for the error to scale inverse-exponentially in the number of ports, both for the deterministic and the probabilistic PBSP. Moreover, given a $d$ dimensional input state, when $N$ independent maximally entangled states are taken as resource states, we design measurements in~\Cref{sec:pbsp} that achieve the same success probability $p$ and {worst-case fidelity $F_\mathrm{wc}$},
\begin{equation}
    p = {F_\mathrm{wc}} = 1-\left(1-\frac{1}{d}\right)^N,
\end{equation} 
for the probabilistic and deterministic PBSP respectively. Moreover, in~\Cref{sec:optimality} we show that this is the optimal achievable error for both the probabilistic and deterministic PBSP when maximally entangled states are used as resource. 

Motivated by the multiple applications of PBT, we also introduce and study a known-state variant of a Universal Programmable Processor~\cite{nielsen_programmable_1997}, which we call Universal Programmable Hybrid Processor (UPHP); a quantum machine that given a description of a quantum state, outputs a state approximately close to the state with a chosen unitary applied to it. The name is meant to evoke a combination between classical and quantum universal processors; classical computers can compute any function on input data by programming a universal gate array, and similarly we can program quantum gate arrays to perform arbitrary unitary operations on quantum data. 

Since there are no correction operations, the PBSP protocol allows the receiver to apply any quantum operation in their registers before receiving the classical message from the sender. We show in~\Cref{thm:construction_uphp} that the memory dimension~$m$ needed for the UPHP processor build from PBSP scales in terms of relative error~$\eps$ as:
\begin{equation}
    {m\leq\left(\frac{1}{\eps}\right)^{4d\ln(d)}.}
\end{equation}
Once again, this shows an exponential separation between the classical-quantum hybrid and fully-quantum universal processors. Finally, in~\Cref{thm:optimal_uphp} we bound the optimal achievable memory dimension for UPHPs by
\begin{equation}
    m\geq 2^{\frac{(1-h(\eps))}{2}d},
\end{equation}
which follows from the impossibility of Quantum Random Access Codes (QRACs) coming from Nayak’s bound, see~\Cref{thm:nayak}. 

\begin{table}
	\begin{center}
		\begin{tabular}{ c|c c|c c| }
			\cline{2-5}
			& \multicolumn{2}{c|}{Existence} & \multicolumn{2}{c|}{Optimality} \\ \cline{2-5}
            & \multicolumn{4}{c|}{PBT} \\ \cline{2-5}
			Det. & ${f\geq 1-\frac{d(d-1)}{N}}$ &\cite{ishizaka_asymptotic_2008} & ${f^*\leq1-\frac{d(d-1)}{8N^2}+O\left(\frac{1}{N^3}\right)}$ &\cite{christandl_asymptotic_2021} \\
			&& & ${f^*}\leq K\frac{\log(d)}{d}\left(2N+\frac{2}{3}\right)$ &\cite{kubicki_resource_2019} \\
			Prob. & $p=\frac{N}{N-1+d^2}$ &\cite{studzinski_port-based_2017} & $p^*\leq \frac{N}{N-1+d^2}$ &\cite{pitalua-garcia_deduction_2013,mozrzymas_optimal_2018}\\ \cline{2-5}
            & \multicolumn{4}{c|}{PBSP} \\ \cline{2-5}
            Det. & {$F_\mathrm{wc}=1-\left(1-\frac{1}{d}\right)^N$} & \Cref{thm:fid_existence_pbsp} & {$1-h(\sqrt{1-F_\mathrm{wc}^*})\leq\frac{\log(d)}{d}4N$} & \Cref{thm:optimal_fidelity_pbsp} \\
            & & & {$F^*_\mathrm{wc,EPR}\leq1-\left(1-\frac{1}{d}\right)^N$} & \Cref{thm:optimal_fidelity_pbsp_2} \\
            Prob. & {$p=1-\left(1-\frac{1}{d}\right)^N$} & \Cref{thm:prob_existence_pbsp} & {$p^*_\mathrm{EPR}\leq1-\left(1-\frac{1}{d}\right)^N$} & \Cref{thm:optimal_prob_pbsp} \\ \cline{2-5}
		\end{tabular}
		\caption{Probability $p$, average fidelity $f$ (see~\Cref{def:ac_fid}) and worst-case fidelity (see~\Cref{def:wc_fid}) of teleporting a qudit with~$N$ ports for PBT and PBSP.}\label{fig:comparison_pbt_pbsp}
	\end{center}
\end{table}

\paragraph{Relation to the literature.}
Our results show a fundamental difference between having an unknown versus known state on Alice’s side. The PBSP protocol achieves an exponential scaling of the error in terms of the number of ports, in contrast to the PBT results, where linear upper bounds exist for optimal scaling for any resource state, see~\Cref{fig:comparison_pbt_pbsp}. For deterministic PBT it is known that the Pretty Good Measurement (PGM) achieves the optimal fidelity both for maximally entangled resource states and when the resource states are optimized~\cite{leditzky_optimality_2022}, but this is not the case for probabilistic PBT~\cite{pitalua-garcia_deduction_2013,mozrzymas_optimal_2018}. Although probabilistic and deterministic PBT behave differently not only in terms of achievable error but also the measurements involved, {for both PBSP tasks we obtain exponentially decreasing error with maximally mixed resource states and practically the same measurement}. We include a summary of the comparison between PBT and PBSP results in~\Cref{fig:comparison_pbt_pbsp}, note that the optimality for PBSP is only discussed for maximally entangled resources, which we denote by an $EPR$ subscript.

Along with the definition of UPQPs Nielsen and Chuang proved that they are impossible to perform perfectly~\cite{nielsen_programmable_1997}, and many have afterwards worked on closing the gap between existence and impossibility results from different directions; e.g.\ studying the possible overlap between two program states~\cite{hillery_approximate_2006}, through epsilon-nets~\cite{majenz_entropy_2018}, or with type constants of Banach spaces~\cite{kubicki_resource_2019}. Finally, a recent result by Yang, Renner and Chiribella~\cite{yang_optimal_2020} closed the gap with a tight bound on the memory dimension, which scales $d^2$-inverse-exponentially with the error, where $d$ is the dimension of the unitary operation. In this regard, the existence result of UPHP we provide scales $d\log(d)$-inverse-exponentially with the error, implying that UPHPs are an exponentially easier task to perform than UPQPs. Moreover, our lower bound obtained using the well known Nayak's bound for QRACs also works for UPQPs and although weaker than the optimal one, it improves on the previous best known one by Kubicki, Palazuelos and Perez-García~\cite{kubicki_resource_2019}. We include a summary of the comparison between UPQPs and UPHPs results in~\Cref{fig:comparison_upqp_uphp}.

\begin{table}
	\begin{center}
		\begin{tabular}{ c|c c|c c| }
			\cline{2-5}
			& \multicolumn{2}{c|}{Upper bounds} & \multicolumn{2}{c|}{Lower bounds} \\ \hline
			UPQP & $\left(\frac{1}{\eps}\right)^{d^2}$ & \cite{kubicki_resource_2019} & $2^{\frac{(1-\eps)}{12}d-\frac{2}{3}\log(d)}$ & \cite{kubicki_resource_2019}\\
            & & & $\frac{(1-\eps)}{4}d$ {if unitary processor}&  \\
			& $\left(\frac{1}{\eps}\right)^{\frac{d^2-1}{2}}$ & \cite{yang_optimal_2020} & $\left(\frac{1}{\eps}\right)^{\frac{d^2-1}{2}-\alpha}$, $\forall \alpha>0$ & \cite{yang_optimal_2020} \\ \hline
            UPHP & {$\left(\frac{1}{\eps}\right)^{4d\ln(d)}$} & \Cref{thm:construction_uphp} & {$2^{\frac{d}{2}(1-h(2\eps))}$} & \Cref{thm:optimal_uphp} \\ \hline
		\end{tabular}
		\caption{Memory dimension~$m$ necessary for universally programming any unitary in $d$ dimensions, the set $\U(d)$, with UPQPs and UPHPs up to error $\eps$.}\label{fig:comparison_upqp_uphp}
	\end{center}
\end{table}

\paragraph{Open questions.}
Since Port-Based Teleportation has attracted so much attention and found so many applications, it is natural to ask which of these are relevant with the knowledge of the state {on} Alice's side, and if they become easier to perform. We include some questions explicitly:
\begin{enumerate}
    \item Exact formulas for the fully optimized measurement and resource exist for both probabilistic and deterministic PBT derived using graphical algebra~\cite{studzinski_port-based_2017,mozrzymas_optimal_2018}, and asymptotic term expressions derived from these are discussed in~\cite{christandl_asymptotic_2021,pitalua-garcia_deduction_2013,mozrzymas_optimal_2018}. Here we only prove optimal achievable error of PBSP when maximally entangled states are used as a resource, but can such an analysis for the fully-optimized case be performed for PBSP?
    \item Recently some efficient algorithms for computing PBT were presented which exploit the symmetry of the task~\cite{grinko_efficient_2023,fei_efficient_2023}. Can we also exploit the symmetry of PBSP to provide efficient algorithms for its implementation?
    \item Note that do not close the gap between the upper and lower bound in the memory dimension for UPHP. In fact, even though we produce scaling in the dimension which is close to tight, we were unable to reproduce the known lower-bound results for UPQP which show that the dimension blows up as the error $\eps$ tends to zero, which should intuitively hold for UPHP as well. 
    \item From the perspective of high-order physics, it turns out that for probabilistic exact unitary transformations, \cite{quintino_probabilistic_2019} showed that the maximal success probability can always be obtained by probabilistic PBT, whilst this is not the case for the deterministic case. What happens in the case of a known state?
\end{enumerate}

\section{Preliminaries}

\subsection{Notation}
We use $\mathcal{D}(\H_d)$ to denote {states in a $d$-level system}, called `qudits', i.e.\ trace-one positive semidefinite matrices in a $d$-dimensional Hilbert spaces
\begin{equation}
    \mathcal{D}(\H_d)\eqdef\{\rho\in PSD(\H)\colon \Tr[\rho]=1\}.
\end{equation}
The single-use distinguishability between two quantum channels is often quantified by the diamond norm, given $\Phi:\Linear(\H_A)\to\Linear(\H_B)$,
\begin{equation}
    \|\Phi\|_\diamond\eqdef \sup_{\substack{k\geq1\\\rho\in\mathcal{D}(\mathbb{C}^k\otimes\H_A)}}\|(I_{\mathbb{C}^k}\otimes \Phi_{A\to B})\rho\|_1.
\end{equation}
As is standard in computer science, we will use $\log$ to denote the base $2$ logarithm. The binary entropy function is denoted by $h$, i.e.\ $h(p)\eqdef -p\log(p)-(1-p)\log(1-p)$. We call `maximally entangled' qudits to
\begin{equation}
    \ket{\phi^+_d} \eqdef \frac{1}{\sqrt{d}}\sum_{x=1}^d \ket{x}\otimes\ket{x}.
\end{equation}
We denote the set of integers from $1$ to $n\in\mathbb{N}$ by $[n]\eqdef\{1,2,\ldots,n\}$. Given $A\eqdef A_1\ldots A_n$, and $x\in[N]$, we denote by $A_{\bar{x}}$ all the elements in $A$ other than the $x$-th, this is, $A_{\bar{x}}\eqdef A_1,\ldots,A_{x-1},A_{x+1},\ldots, A_N$.

While we denote a pure quantum state by $\ket{\psi}$, we use $\psi$ to denote its classical description.
Because our results are of information-theoretic nature, we will assume the classical description to be of arbitrarily-high precision representation of $\ket{\psi}$ unless specified otherwise.

We denote by $F(\rho,\sigma)$ the fidelity between two states 
\begin{equation}
    F(\rho,\sigma)=\norm{\sqrt{\rho}\sqrt{\sigma}}^2_1=\Tr[\sqrt{\sqrt{\rho}\sigma\sqrt{\rho}}]^2.
\end{equation}
If one of the states is pure, e.g.\ $\rho=\ketbra{\psi}$, then the fidelity is the overlap of the states
\begin{equation}
    F(\rho,\sigma)=\Tr[\sqrt{\ketbra{\psi}\sigma\ketbra{\psi}}]^2=\bra{\psi}\sigma\ket{\psi}.
\end{equation}
Moreover, by the Fuchs-van de Graaf inequalities we have the following relation between the fidelity and distance between two states,
\begin{equation}\label{eq:fvdg}
    1-\sqrt{F(\rho,\sigma)}\leq\frac{1}{2}\norm{\rho-\sigma}_1\leq\sqrt{1-F(\rho,\sigma)}.
\end{equation}

\begin{definition}[Entanglement fidelity]\label{def:entanglement_fid}
    Given two channels $\Phi^1_A,\Phi^2_A\in CPTP(\H_A)$, we define their \emph{entanglement fidelity} as
    \begin{equation}
        F(\Phi^1_A,\Phi^2_A)\eqdef F((\Phi^1_A\otimes I_B)\ketbra{\phi^+_{AB}},(\Phi^2_A\otimes I_B)\ketbra{\phi^+_{AB}}).
    \end{equation}
\end{definition}
\begin{definition}[Worst-case fidelity]\label{def:wc_fid}
    Given two channels $\Phi^1_A,\Phi^2_A\in CPTP(\H_A)$, we define their \emph{worst-case fidelity} as
    \begin{equation}
        F_\mathrm{wc}(\Phi^1_A,\Phi^2_A)\eqdef \inf_{\ket{\psi}}F((\Phi^1_A\otimes I_B)\ketbra{\psi_{AB}},(\Phi^2_A\otimes I_B)\ketbra{\psi_{AB}}).
    \end{equation}
\end{definition}
We denote by $F^*(\Phi,\mathcal{I})$ (respectively $F_\mathrm{wc}^*(\Phi,\mathcal{I})$) the optimal achievable entanglement (respectively worst-case) fidelity between a channel $\Phi$ and the identity channel $\mathcal{I}$, we will call this the entanglement (respectively worst-case) fidelity of the channel $\Phi$ and denote it by $F^*$ (respectively $F_\mathrm{wc}^*$) if the protocol we are talking about is clear from the context.

\begin{definition}[Average fidelity]\label{def:ac_fid}
    Given a channel $\Phi\in CPTP(\H_A)$, we define its \emph{average fidelity} as
    \begin{equation}
        f(\Phi)\eqdef \int_\psi\bra{\psi}\Phi(\ketbra{\psi})\ket{\psi}\mathrm{d}\psi.
    \end{equation}
\end{definition}
We denote by $f^*(\Phi)$ the optimal achievable average of the channel $\Phi$, and denote it by $f^*$ if the protocol we are talking about is clear from the context. There is a well-known relation between average fidelity and entanglement fidelity of any channel, see~\cite[Proposition 1]{horodecki_general_1999}:
\begin{equation}
    f(\Phi)=\frac{F(\Phi,\mathcal{I})d+1}{d+1}.
\end{equation}

\subsection{Quantum Random Access Codes}

A fundamental impossibility result of quantum information theory is Holevo's theorem, i.e., no more than $n$ bits of \emph{accessible} information can be transmitted by transferring $n$ qubits. Intuitively, Holevo's bound comes from the necessity of choosing a measurement, which will possibly destroy the information that could have been obtained by another measurement. Quantum Random Access Codes (QRACs) were introduced by Ambainis, Nayak, Ta-Shma, and Vazirani~\cite{ambainis_dense_1999} with the hope of circumventing Holevo's theorem. The idea was to hand some measurements to the receiver -- as many measurements as bits -- such that each measurement would reveal at least one of the bits with a certain fixed probability of success. Note that the existence of QRACs would not immediately violate Holevo's bound: as the measurements do not commute, they cannot be applied subsequently with the same success probability.

\begin{definition}[QRAC] A $(n,m,p)$-quantum random access code (QRAC) is a mapping of $n$ classical bits into $m$ qubits $(f:\{0,1\}^n\to \C^{2^m},\: x\mapsto \rho_x)$, along with a set of measurements $(M_1^0,M_1^1),\ldots,(M_n^0,M_n^1)$ satisfying, for all $x\in\{0,1\}^n$ and $i\in[n]$,
	\[ \Tr(M_i^{x_i}f(x))\geq p. \]
\end{definition}

Unfortunately, QRACs do not improve much on Holevo's bound, as shown in~\cite{nayak_optimal_1999}. In particular, although an $n$-qubit quantum state is a vector in a $2^n$-dimensional complex space, it -- perhaps surprisingly -- cannot encode an amount of information that is more than linear in $n$; or at least not in an accessible way.
\begin{theorem}[Nayak's bound]\label{thm:nayak}
	Any $(n,m,p)$-QRAC must satisfy
	\[ m\geq n(1-h(p)), \]
    where $h$ is the binary entropy function.
\end{theorem}

\subsection{Port-Based Teleportation}

On the other side of the story, we can ask if the converse of such a statement is true, that is, what the number of bits we need to specify a qubit is. Intuitively, since we can only approximate a complex number with a finite number of bits, this should be impossible, but teleportation tasks circumvent this limitation by making use of entanglement as a resource.

Here we are interested in a particular teleportation task called Port-Based Teleportation (PBT),  introduced by Ishizaka and Hiroshima~\cite{ishizaka_asymptotic_2008}, where Bob's computational power is restricted to tracing out registers. Note that in this situation the only meaningful classical communication {is communicating a port}, thus the classical communication cost is completely determined by the number of registers that Bob holds, in other words, the dimension of Bob's part of the resource state. 
\begin{algorithm}
    \floatname{algorithm}{Protocol}
    \caption{$(d,N,p)$ Probabilistic Port-Based Teleportation}\label{alg:prob_pbt}
    \textit{Setting.} Both Alice and Bob each have $N$ qudits in $A_1,A_2,\ldots, A_N$ and $B_1,B_2,\ldots, B_N$ respectively, which are denoted by $A$ and $B$ as a whole. They share an arbitrary entangled resource in $AB$, and Alice holds an additional input register $A_0$. There is a single round of classical communication from Alice to Bob. Bob's correction operations are restricted to tracing out registers.
    \sbline
    \textit{Inputs.} A $d$-dimensional quantum state in Alice's register $A_0$.
    \sbline
    \textit{Goal.} Alice's input state to be teleported to an arbitrary output port of Bob, with probability of success $p$.
    \sbline
    \textit{The protocol:}
    \begin{algorithmic}[1]
		\STATE Alice performs a $N+1$ outcome POVM measurement on her $d$-dimensional input state and her half of the $N$ shared resource states.
        \STATE Alice communicates either a port $x\in[N]$ or abort $x=0$ to Bob with $\log(N+1)$ bits of classical communication.
		\STATE If $x=0$ Bob aborts, else he traces out all but register $B_x$.
    \end{algorithmic}
\end{algorithm}

We can distinguish two types of PBT tasks; probabilistic and deterministic ones. In the probabilistic case, Alice performs a $N+1$ outcome POVM $\mathcal{M}=\{M_x\}_{x=0}^N$ such that only outcomes $x=1,\ldots,N$ refer to a successful transmission of the state, in this case the state in port $x$ has perfect fidelity with respect to the state intended, and outcome $x=0$ refers to a failure in the transmission, therefore Bob aborts. In this setting it is natural to study the performance of the protocol by studying the success probability. To avoid confusion, we will denote these as $(d,N,p)$-PBT protocols, where $p$ denotes the probability of success of the protocol.

Deterministic PBT does not consider the option of failure, instead Alice performs a $N$ outcome POVM, and always transmits a port. It is clear that such a protocol must be approximate, meaning that the outcome state cannot always have perfect fidelity with the intended one, hence we study the performance of the protocol in terms of the {entanglement fidelity of the channel}.
\begin{algorithm}
    \floatname{algorithm}{Protocol}
    \caption{$(d,N,F)$ Deterministic Port-Based Teleportation}\label{alg:det_pbt}
    \textit{Setting.} Both Alice and Bob each have $N$ qudits in $A_1,A_2,\ldots, A_N$ and $B_1,B_2,\ldots, B_N$ respectively, which are denoted by $A$ and $B$ as a whole. They share an arbitrary entangled resource in $AB$, and Alice holds an additional input register $A_0$. There is a single round of classical communication from Alice to Bob. Bob's correction operations are restricted to tracing out registers.
    \sbline
    \textit{Inputs.} A $d$-dimensional quantum state in Alice's register $A_0$.
    \sbline
    \textit{Goal.} Alice's input state to be teleported to an arbitrary output port of Bob, with fidelity $F$.
    \sbline
    \textit{The protocol:}
    \begin{algorithmic}[1]
		\STATE Alice performs a $N$ outcome POVM measurement on her $d$-dimensional input state and her half of the $N$ shared resource states.
        \STATE Alice always communicates a port $x\in[N]$ to Bob with $\log(N)$ bits of classical communication.
		\STATE Bob traces out all but register $B_x$.
    \end{algorithmic}
\end{algorithm}

For protocols with a covariance property, {as is the case of PBT}, both the entanglement fidelity $F$ and the diamond norm $\eps$ are equivalent~\cite[Appendix B, Lemma 2]{yang_optimal_2020},
\begin{equation}
    F = 1-\frac{\eps}{2},
\end{equation}
In the rest of this work we will talk about probability of success and entanglement fidelity as measures of accuracy.

In particular, any probabilistic $(d,N,p)$-PBT can be used to construct a deterministic one, just by Alice transmitting a random register index whenever the measurement outcome is failure, i.e.\ $x=0$. The achieved entanglement fidelity is 
\begin{equation}\label{eq:fidelity_from_prob}
    F\geq p+\frac{1-p}{d^2}.
\end{equation}

There are some known trade-offs between the dimension of the input state $d$, the number of ports $N$, and the accuracy of the protocol, depending on the resource states, for both the probabilistic and the deterministic PBT. As all these parameters can be optimized independently, in order to simplify the analysis it is standard in the literature to consider $N$ independent maximally entangled resource states (one per port) and the Pretty Good Measurement (PGM), in this case we call the protocol \textit{standard}. We include the known results between resource trade-offs (as far as we are aware) in~\Cref{fig:comparison_pbt}. 

\begin{table}[ht]
	\begin{center}
		\begin{tabular}{ c|c c|c c| }
			\cline{2-5}
			& \multicolumn{2}{c|}{Existence} & \multicolumn{2}{c|}{Optimality} \\ \hline
			Det. & $F^{\mathrm{sta}}\geq 1-\frac{d^2-1}{N}$ &\cite{ishizaka_asymptotic_2008}& $F^*\leq1-\frac{1}{4(d-1)N^2}+O\left(\frac{1}{N^3}\right)$ &\cite{ishizaka_remarks_2015} \\ \cline{2-5}
			& $\eps^{\mathrm{sta}}\leq \frac{4d^2}{\sqrt{N}}$ &\cite{beigi_simplified_2011} & $F^*\leq1-\frac{d^2-1}{8N^2}+O\left(\frac{1}{N^3}\right)$ &\cite{christandl_asymptotic_2021} \\ \cline{2-5}
			&$F\geq1-\frac{d^2(d^2-2)-1}{d^2(d^2+N-1)}$& \cref{eq:fidelity_from_prob}& $F^*\leq K\frac{\log(d)}{d}\left(2N+\frac{2}{3}\right)$ &\cite{kubicki_resource_2019} \\ \hline
			Prob. & $p=\frac{N}{N-1+d^2}$ &\cite{studzinski_port-based_2017} & $p^*\leq \frac{N}{N-1+d^2}$ &\cite{pitalua-garcia_deduction_2013,mozrzymas_optimal_2018}\\ \cline{2-5}
		\end{tabular}
		\caption{Probability $p$ and entanglement fidelity $F$ (see~\Cref{def:entanglement_fid}) of teleporting a qudit with~$N$ ports for PBT.}\label{fig:comparison_pbt}
	\end{center}
\end{table}

\subsection{Universal Programmable Quantum Processors}

An interesting application of PBT is of Universal Programmable Quantum Processors (UPQPs). We ideally want to build quantum computers as we do with the classical ones, by building a universal processor independent of the data and programs that we insert, only that now both these elements can be quantum states or unitary. Informally the task is the following: Given a unitary, we are asked to apply it in an unknown state, but only after losing access to the unitaries. Therefore, we want to find a way to store the unitary in a quantum memory. 

\begin{definition}[UPQP]
	A channel $\P\in CPTP(\H_D \otimes \H_M)$ is a $d$-dimensional Universal Programmable Quantum Processor if for every unitary $U\in\U(\H_D)$, there exists a unit vector $\ket{\phi_U}\in \mathcal{D}(\H_M)$ such that
	\begin{equation}
        \Tr_M[\P(\rho_D\otimes \dyad{\phi_U}_M)] = U\rho_D U^\dagger,\quad \text{for all  } \rho_D\in \mathcal{D}(\H_D).
    \end{equation}
\end{definition}

Along their introduction, Nielsen and Chuang~\cite{nielsen_programmable_1997} proved the {impossibility of perfectly} implementing UPQPs. The basic idea is that for every unitary we desire to apply, there needs to {be an orthogonal} memory state, and since for any given dimension there {are infinitely many possible} unitaries, we need an infinite dimensional memory space to store them. However, the proof does not hold if we relax the previous definition; that is, for approximate UPQPs, where we only want to implement the desired unitary with approximate accuracy (or with certain probability of success).

\begin{table}
	\begin{center}
		\begin{tabular}{ c|c c|c c| }
			\cline{2-5}
			& \multicolumn{2}{c|}{Upper bounds} & \multicolumn{2}{c|}{Lower bounds} \\ \hline
			Error dep. & $\left(\frac{1}{\eps}\right)^{d^2}$ & \cite{kubicki_resource_2019} & $\frac{d^2}{\eps}$ & \cite{majenz_entropy_2018}\\ \cline{2-5}
			& $\left(\frac{1}{\eps}\right)^{\frac{d^2-1}{2}}$ & \cite{yang_optimal_2020} & $\left(\frac{1}{\eps}\right)^{\frac{d^2-1}{2}-\alpha}$, $\forall \alpha>0$ & \cite{yang_optimal_2020} \\ \hline
			Dimension dep. & $d^2\log(1/\eps)$ & \cite{kubicki_resource_2019} & $\frac{(1-\eps)}{12}d-\frac{2}{3}\log(d)$ & \cite{kubicki_resource_2019} \\
            $\log$ power & & & $\frac{(1-\eps)}{4}d$ {if unitary processor} &  \\ \cline{2-5}
            & $d^2\frac{\log(1/\eps)}{2}$ & \cite{yang_optimal_2020} & $d^2\frac{\log(1/\eps)}{2}-\beta$, $\forall\beta>\frac{\log(1/\eps)}{2}$ & \cite{yang_optimal_2020} \\ \cline{2-5}
		\end{tabular}
		\caption{Performance of $d$-dimensional $\eps$-UPQP in terms of memory dimension $m$.}\label{fig:comparison_upqp}
	\end{center}
\end{table}

\begin{definition}[Approximate UPQP] A channel $\P\in CPTP(\H_D\otimes \H_M)$ is a $d$-dimensional $\eps$-Universal Programmable Quantum Processor, if for every unitary $U\in\U(\H_D)$ there exists a unit vector $\ket{\phi_U}\in \mathcal{D}(\H_M)$ such that
	\begin{equation}\label{eq:aUPQP} 
    \frac{1}{2}\norm{\Tr_M\left[\P(\cdot\otimes\dyad{\phi_U}_M)\right]-U(\cdot)U^\dag}_\diamond\leq\eps.
	\end{equation}
\end{definition}

However, we cannot do much better for approximate UPQPs. 
Yang, Renner, and Chiribella~\cite{yang_optimal_2020} showed tight upper and lower bounds for the memory dimension that both blows up inverse polynomially with small error $\eps$, and exponentially when increasing the dimension of the input state $d$, closing a long-lasting open question. 
We include the known results between resource trade-offs (as far as we are aware) in~\Cref{fig:comparison_upqp}.

\section{Port-Based State Preparation}\label{sec:pbsp}

Here we introduce a classical-quantum hybrid communication task where Alice holds a complete classical description of a quantum state, and her goal is for Bob to end up with the quantum state whose description she holds in one of his ports. In order to achieve this, both parties share a finite amount of entangled resources, Alice performs a {measurement on her part} of the shared resources and Bob's actions are restricted to tracing out his side of the resources, or `ports'. We call this Port-Based State Preparation (PBSP).

In a way, the task is a hybrid between Remote State Preparation (RSP) tasks introduced in the seminal paper by Bennett, DiVicenzo, Shor, Smolin, Terhal, and Wootters~\cite{bennett_remote_2001}, and Port-Based Teleportation. Recall that RSP achieves asymptotic rates of one qubit transmission per bit by allowing Bob to perform any operation on his side, and we are interested in knowing what happens if we restrict Bob to covariant classical operations, as in the case of PBT.

\begin{figure}[ht]
    \centering
    \tikzset{every picture/.style={line width=0.75pt}} 

\begin{tikzpicture}[x=0.75pt,y=0.75pt,yscale=-1,xscale=1]

\draw    (109,180) -- (89,120) -- (109,60) ;
\draw    (109,60) -- (129,60) ;
\draw    (109,180) -- (237,180) ;
\draw    (109,199.97) -- (89,139.97) -- (109,79.97) ;
\draw    (109,79.97) -- (129,79.97) ;
\draw    (109,200) -- (237,200) ;
\draw   (129,90) .. controls (129,90) and (129,90) .. (129,90) .. controls (162.14,90) and (189,69.85) .. (189,45) .. controls (189,20.15) and (162.14,0) .. (129,0) -- cycle ;
\draw    (89,10) -- (129,10) ;
\draw    (89,14) -- (129,14) ;
\draw   (189,43) -- (259,43) -- (259,170) ;
\draw   (189,47) -- (255,47) -- (255,170) ;
\draw   (237,170) -- (277,170) -- (277,210) -- (237,210) -- cycle ;
\draw    (277,190) -- (297,190) ;
\draw [color={rgb, 255:red, 0; green, 0; blue, 0 }  ,draw opacity=0.5 ] [dash pattern={on 4.5pt off 4.5pt}]  (0,130) -- (380,130) ;

\draw (51,95.4) node [anchor=north west][inner sep=0.75pt]    {$ \begin{array}{l}
\ket{\phi _{d}^{+}}\\
\ \ \vdots \\
\ket{\phi _{d}^{+}}
\end{array}$};
\draw (146,38.4) node [anchor=north west][inner sep=0.75pt]    {$\mathcal{M}$};
\draw (60,3.4) node [anchor=north west][inner sep=0.75pt]    {$\psi _{d}$};
\draw (240,180.4) node [anchor=north west][inner sep=0.75pt]    {$Tr_{B_{\overline{x}}}$};
\draw (272,73.4) node [anchor=north west][inner sep=0.75pt]    {$x\in \{0,1\}^{*}$};
\draw (315,181.4) node [anchor=north west][inner sep=0.75pt]    {$=\ \ket{\psi _{d}}$};
\draw (2,38) node [anchor=north west][inner sep=0.75pt]   [align=left] {Alice};
\draw (2,180) node [anchor=north west][inner sep=0.75pt]   [align=left] {Bob};

\end{tikzpicture}
    \caption{Port-Based State Preparation protocol.}
\end{figure}
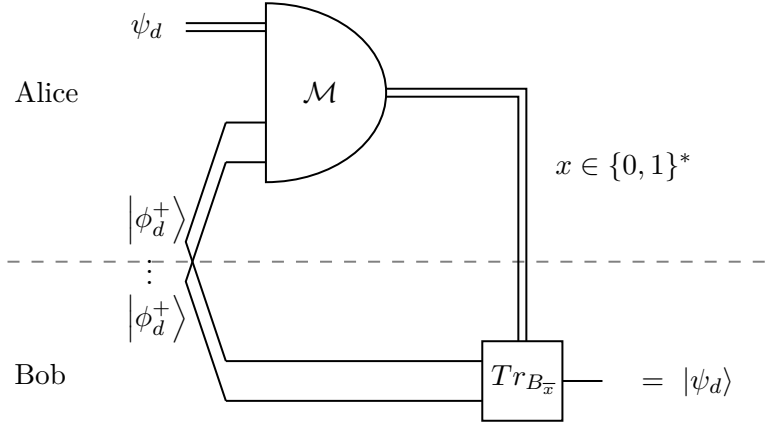

We can distinguish between probabilistic and deterministic PBSP tasks. Assuming that both parties share $N$ entangled states, in probabilistic PBSP~(\Cref{alg:prob_pbsp}), Alice's measurement has $N+1$ possible outcomes; one per port, after which Bob's register will be exactly in the desired state, and one for the failure probability. Naturally we use the protocol's success probability as an accuracy measure.

\begin{algorithm}
    \floatname{algorithm}{Protocol}
    \caption{$(d,N,p)$ Probabilistic Port-Based State Preparation.}\label{alg:prob_pbsp}
    \textit{Setting.} {Both Alice and Bob each have $N$ qudits in $A_1,A_2,\ldots, A_N$ and $B_1,B_2,\ldots, B_N$ respectively, which are denoted by $A$ and $B$ as a whole. They share an arbitrary entangled resource in $AB$}. There is a single round of classical communication from Alice to Bob. Bob's correction operations are restricted to tracing out registers.
    \sbline
    \textit{Inputs.} Alice receives the classical description of a $d$-dimensional quantum state.
    \sbline
    \textit{Goal.} Alice's input state to be prepared in an arbitrary output port of Bob, with probability of success $p$.
    \sbline
    \textit{The protocol:}
    \begin{algorithmic}[1]
		\STATE Alice performs a $N+1$ outcome POVM measurement on her half of the $N$ shared resource states, which can depend on the classical description of her state.
        \STATE Alice communicates either a port $x\in[N]$ or abort $x=0$ to Bob with $\log(N+1)$ bits of classical communication.
		\STATE If $x=0$ Bob aborts. Else he traces out all but register $B_x$.
    \end{algorithmic}
\end{algorithm}

In deterministic PBSP~(\Cref{alg:det_pbsp}), Alice's measurement has $N$ possible outcomes, one per port, but in this case it is impossible for Bob's register after tracing out to be exactly Alice's intended state for all the inputs. In this case, we use the fidelity of the outcome state with respect to the input as a measure of success of the protocol. 

\begin{remark}
    Conventional channel fidelity definitions consider purifications of a given state with respect to a reference space and apply the channel to half of it, as is the case in~\Cref{def:entanglement_fid,def:wc_fid}. Nevertheless the input of our PBSP protocol is a classical description of the state we intend to send, thus it is not clear what it would mean to apply our classical-to-quantum map to half of it. We could instead consider average fidelity as a figure of merit, see~\Cref{def:ac_fid}, this is 
    \begin{equation}
        f(\Phi)\eqdef\int_\psi\bra{\psi}\Phi(\psi)\ket{\psi}\mathrm{d}\psi.
    \end{equation}
    However, this integral is not necessarily well defined. Instead, we define the \emph{worst-case fidelity} of a classical-to-quantum map as
    \begin{equation}
        F_\mathrm{wc}(\Phi)\eqdef \inf_{\ket{\psi}\in\H}F(\Phi(\psi),\ket{\psi}).
    \end{equation}

    Intuitively the covariance-like property of PBSP makes the worst-case fidelity comparable to average fidelity. Clearly, $F_\mathrm{wc}(\Phi)\leq f(\Phi)$. Moreover, given a map $\Phi:\{0,1\}^*\to\mathcal{H}_d$, there exists a map $\hat{\Phi}:\{0,1\}^*\to\mathcal{H}_d$ such that $F_\mathrm{wc}(\hat{\Phi})\geq f(\Phi)$. For example, we can define the map $\hat{\Phi}$ by Alice and Bob first agreeing to a random $d$-dimensional unitary $U\sim\mu_d$ sampled with respect to the Haar-measure, then performing the protocol $\psi\mapsto \Phi(U\psi U^\dag)$, and finally Bob uncomputing the unitary $U$. The new protocol has the desired fidelity
    \begin{equation}
        F_\mathrm{wc}(\hat{\Phi})
        =\inf_{\ket{\psi}\in\H_d}\int_U \bra{\psi}U^\dag\Phi(U\psi U^\dag)U\ket{\psi}\mathrm{d} U
        =\inf_{\ket{\psi}\in\H_d}\int_\phi \bra{\phi}\Phi(\phi)\ket{\phi}\mathrm{d} \phi=f(\Phi).
   \end{equation}
   We stress again that this is only an intuitive check to justify the use of the worst-case fidelity as a figure of merit, but the above integrals over classical-to-quantum maps are not well defined.
\end{remark}

\begin{algorithm}
    \floatname{algorithm}{Protocol}
    \caption{$(d,N,F)$ Deterministic Port-Based State Preparation.}\label{alg:det_pbsp}
    \textit{Setting.} {Both Alice and Bob each have $N$ qudits in $A_1,A_2,\ldots, A_N$ and $B_1,B_2,\ldots, B_N$ respectively, which are denoted by $A$ and $B$ as a whole. They share an arbitrary entangled resource in $AB$}. There is a single round of classical communication from Alice to Bob. Bob's correction operations are restricted to tracing out registers.
    \sbline
    \textit{Inputs.} Alice receives the classical description of a $d$-dimensional quantum state.
    \sbline
    \textit{Goal.} Alice's input state to be prepared in an arbitrary output port of Bob, with fidelity $F$.
    \sbline
    \textit{The protocol:}
    \begin{algorithmic}[1]
		\STATE Alice performs a $N$ outcome POVM measurement on her half of the $N$ shared resource states, which can depend on the classical description of her state.
        \STATE Alice always communicates a port $x\in[N]$ to Bob with $\log(N)$ bits of classical communication.
		\STATE Bob traces out all but register $B_x$.
    \end{algorithmic}
\end{algorithm}

As mentioned earlier, Alice holding the complete classical description of the desired state allows her to construct a measurement that depends on it. In particular, this allows her to perform independent measurement operations in each port, which turns out to be enough to obtain a probability of success that exponentially increases in terms of the number of ports. Intuitively, this is equivalent to independent coin tossing, something that we are not allowed in the original Port-Based Teleportation tasks.

\begin{theorem}\label{thm:prob_existence_pbsp}
	There exists a $(d,N,p)$ probabilistic PBSP protocol with probability of success 
	\begin{equation}
		p=1-\left(1-\frac{1}{d}\right)^N.
	\end{equation}
\end{theorem}
\begin{proof}
	Let us assume that Alice and Bob share $N$ independent maximally entangled qudits as a resource. Given a description of a $d$-dimensional state $\ket{\psi}$, denoted $\psi$, Alice's measurement will {consist of} checking if the state is in any port. This can be done by just projecting to the conjugate $\ket{\psi^*}$, because for every register $i\in[N]$ we have
    \begin{gather}
        \Tr_{A_i}\left[(\ketbra{\psi^*}_{A_i}\otimes I_{B_i})\ketbra{\phi^+_d}_{A_iB_i}\right] = \frac{1}{d}\ketbra{\psi}_{B_i}, \\
        \Tr_{A_i}\left[((I-\ketbra{\psi^*}_{A_i})\otimes I_{B_i})\ketbra{\phi^+_d}_{A_iB_i}\right] = \frac{1}{d}(I-\ketbra{\psi})_{B_i}.
    \end{gather}
    Alice's partial measurement can be normalized to obtain the POVM $\mathcal{M}=\{M_x(\psi)\}_{x=0}^N$ with elements
	\begin{equation}\label{eq:meas_dpbsp}\begin{split}
		M_0(\psi)& \eqdef \bigotimes_{i=1}^N (I-\ketbra{\psi^*})_{A_i},\\
		M_x(\psi)& \eqdef\sum_{\substack{S \subseteq [N]\\ \text{s.t.} \, x\in S}} \left[ \frac{1}{|S|}
        \bigotimes_{\substack{ i \in S}}\ketbra{\psi^*}_{A_i} \otimes \bigotimes_{\substack{ i \not\in S}}(I-\ketbra{\psi^*})_{A_i}\right],\quad \text{for } x\in[N].
        \end{split}
	\end{equation}
    Outcome $x=0$ means that the state has not been found in any port, and thus the protocol has failed. After any other outcome $x\in[N]$, Bob's system in register $B_x$ will be in the desired state
    \begin{equation}\begin{split}
        \rho_{B_x}&=\frac{\Tr_{AB_{\bar{x}}}\left[M_x(\psi)\bigotimes_{i=1}^N\ketbra{\phi_d^+}_{A_iB_i}\right]}{\Tr[M_x(\psi)\bigotimes_{i=1}^N\ketbra{\phi_d^+}_{A_iB_i}]}\\
        &=d\Tr_{A_x}\left[(\ketbra{\psi^*}_{A_x}\otimes I_{B_x})\ketbra{\phi^+_d}_{A_xB_x}\right]\\
        &=\ketbra{\psi}_{B_x}.
    \end{split}\end{equation}
    Given this measurement, the probability of success of the protocol can be calculated easily as we have a closed expression for the failing measurement
    \begin{equation}
        p = 1-\Tr[M_0(\psi)\bigotimes_{i=1}^N\ketbra{\phi_d^+}_{A_iB_i}] = 1-\left(1-\frac{1}{d}\right)^N.
    \end{equation}
\end{proof}

We can always construct a deterministic PBSP protocol from a probabilistic PBSP by just outputting a random port when a failure is measured. Although in the PBT protocol this construction gives us a greater fidelity than the probability of success, see~\Cref{eq:fidelity_from_prob}, for the PBSP case this is not necessarily true. For example, in the measurement we described above the extra bit is always orthogonal to the desired state, thus the deterministic PBSP obtained from spreading the failure evenly among the ports leads to the same channel fidelity as the probability of success of the probabilistic protocol described in~\Cref{thm:prob_existence_pbsp}.

\begin{theorem}\label{thm:fid_existence_pbsp}
    There exists a $(d,N,F)$ deterministic PBSP protocol with channel fidelity
	\begin{equation}
		{F_\mathrm{wc}}=1-\left(1-\frac{1}{d}\right)^N.
	\end{equation}
\end{theorem}
\begin{proof}
    Let us assume that Alice and Bob share $N$ independent maximally entangled qudits as a resource. Given a description of a $d$-dimensional state $\ket{\psi}$, denoted $\psi$, Alice's measurement will {consist of} checking if the state is in any port via the measurements 
    \begin{equation}\label{eq:meas_fidelity_pbsp}
		M_x'(\psi)= M_x(\psi) + \frac{1}{N}\bigotimes_{i=1}^N(I-\ketbra{\psi^*})_{A_i},\quad\text{for }x\in[N],
	\end{equation}
    where the $M_x(\psi)$'s are the same as in~\Cref{eq:meas_dpbsp}. Any of the outcomes $x\in[N]$ will now occur with equal probability
    \begin{equation}\label{eq:prob_meas}\begin{split}
		p_x & =\Tr[(M_x'(\psi)\otimes I_B)\bigotimes_{i=1}^N\ketbra{\phi_d^+}_{A_iB_i}] \\
        & =\frac{1}{d}\sum_{i=0}^{N-1}\frac{\binom{N-1}{i}}{i+1}\left(\frac{1}{d}\right)^i\left(1-\frac{1}{d}\right)^{(N-1)-i} + \frac{1}{N}\left(1-\frac{1}{d}\right)^N,
    \end{split}
	\end{equation}
    and after outcome $x\in[N]$ the state in Bob's register $B_x$ will be 
	\begin{multline}
        \rho_{B_x} = \frac{1}{p_x}\left(\frac{1}{d}\ketbra{\psi}_{B_x}\sum_{i=0}^{N-1}\frac{\binom{N-1}{i}}{i+1}\left(\frac{1}{d}\right)^{i}\left(1-\frac{1}{d}\right)^{(N-1)-i}\right.\\
        \left.+\frac{1}{Nd}(I-\ketbra{\psi})_{B_x}\left(1-\frac{1}{d}\right)^{N-1}\right).
	\end{multline}
	From the above expression we can see that for every input $\psi$ the fidelity between the outcome of the protocol $\Phi(\psi)$ and the desired state $\ket{\psi}$ will be equivalent to the probability of success $p$, as the latest term is orthogonal to $\ket{\psi}_{B_x}$ and thus will not contribute to the overlap between the states. Formally,
	\begin{equation}\begin{split}
		F\left(\Phi(\psi),\ket{\psi}\right)&=\sum_{x=1}^Np_x\bra{\psi}_{B_x}\rho_{B_x}\ket{\psi}_{B_x}=\frac{N}{d}\sum_{i=0}^{N-1}\frac{\binom{N-1}{i}}{i+1}\left(\frac{1}{d}\right)^{i}\left(1-\frac{1}{d}\right)^{(N-1)-i}\\
		&=1-\left(1-\frac{1}{d}\right)^N.
	\end{split}\end{equation}
\end{proof}

Note that the accuracy of PBSP scales exponentially in terms of the number of ports, both for the probabilistic and the deterministic case. This contrasts the PBT case; where the upper bound for both terms is linear in terms of number of ports. In other words, having the classical description of the state gives the sender enough power to go from linear to exponential accuracy in the number of ports; or that PBSP is a strictly easier task to perform than PBT.

\section{Universal Programmable Hybrid Processors}

Quantum computers have the ability to perform arbitrary unitary operations in two-level systems, i.e., we can decompose any unitary acting on qubits in gate arrays which can be implemented using finite resources. Moreover, classical computers can perform arbitrary operations on any input data with a fixed universal gate array, which can be programmed to perform any operation on the input data.  However, this was shown to be impossible quantumly by Nielsen and Chuang~\cite{nielsen_programmable_1997}: the so-called \emph{no-programming theorem} states that perfect universal quantum processors require an orthogonal state for each unitary that we desire to perform, and since there {are infinitely many unitaries} acting on any $d$-level system, we would need infinite-dimensional spaces to perform arbitrary unitaries on quantum states.
The relaxed scenario of approximate UPQPs has also attracted much attention~\cite{hillery_approximate_2006,kubicki_resource_2019,majenz_entropy_2018}, but a recent article by Yang, Chiribella and Renner~\cite{yang_optimal_2020} closed the gap between the upper and lower bound on the memory dimension $m$ of a $d$-dimensional $\eps$-approximate UPQP:
\begin{equation}
    \left(\frac{1}{\eps}\right)^{\frac{d^2-1}{2}}\geq m \geq \left(\frac{1}{\eps}\right)^{\frac{d^2-1}{2}-\alpha},\quad\text{for all }\alpha>0.
\end{equation}
In other words, the memory dimension scales with the dimension-square exponentially in the error, which nearly saturates the upper bound in~\cite{kubicki_resource_2019} given by the $\eps$-net of the unitaries. Intuitively, it is as saying that the memory cannot do much better than pointing to a unitary from the $\eps$-net, which is in a way a classical operation.

Since for some applications the user of the quantum program might actually know which state they would like to apply the (to them unknown) unitary to, we imagine an analogous task to $\eps$-UPQPs, where we want to program a universal gate array but to which we will give a classical input and expect to perform an arbitrary unitary operation on the quantum state the classical input specifies. The initial state of our classical-quantum system will be of the form
\begin{equation}
    (\psi,\ket{\phi}_M),
\end{equation}
where~$\psi$ will refer to the classical description of the $d$-dimensional quantum state~$\ket{\psi}$ we want to perform the operation on, w.l.o.g.\ we assume it to be pure, and~$\ket{\phi}_M$ is the state of the $m$-dimensional program register. For any input, the processor~$\mathcal{P}$ will then map this system to
\begin{equation}
    (\psi,\ket{\phi}_M)\mapsto \P(\psi,\ket{\phi}_M),
\end{equation}
and we say that the processor $\P$ implements a unitary $U$ in the data register~$\ket{\psi}_D$ specified by the classical input $\psi$ if 
\begin{equation}\label{eq:def_qcupp}
    {\Tr_M\left[\P(\psi,\ket{\phi}_M)\right] = U\ketbra{\psi}_DU^\dag.}
\end{equation}
Note that although the processor itself is not unitary, since it takes classical and quantum data, we want~\Cref{eq:def_qcupp} to hold for \emph{any input state}, thus it is not possible to just encode all the information about the unitary desired in the classical data. Let us denote by~$\P_\psi\eqdef \P(\psi,\cdot)$ the universal gate array specified by classical input~$\psi$. 

As in the case of $\eps$-UPQPs, we are interested in an approximate version of this processor. We will call such a system a \emph{approximate universal programmable hybrid processor} ($\eps$-UPHP) to refer to the classical-input quantum-output dynamics of the system explicitly.
We can also distinguish between probabilistic and deterministic $\eps$-UPHPs, in analogy to the different variants of port-based teleportation, but since every probabilistic processor can be converted into a deterministic one, it is more common to treat the latter. We will do the same here for the sake of simplicity.

\begin{definition}[$\eps$-UPHP] A family of channels $\P\eqdef\{P_{\psi}\}_{\psi}\subset CPTP(\H_M,\H_D\otimes \H_M)$ is a \emph{$d$-dimensional $\eps$-Universal Programmable Hybrid Processor}, if for every unitary $U\in\U(\H_D)$ there exists a unit vector $\ket{\phi_U}\in \mathcal{D}(\H_M)$ such that
\begin{equation}
    {\frac{1}{2}}\norm{\Tr_M\left[\P_\psi(\dyad{\phi_U}_M)\right]-U\dyad{\psi}_DU^\dag}_1\leq\eps\quad\text{for all input }\psi,
\end{equation}
where $\psi$ is the classical description of the $d$-dimensional quantum state $\ket{\psi}_D\in\mathcal{D}(\H_D)$.
\end{definition}
We describe the above accuracy parameter $\eps$ in terms of the {trace norm}, equivalent to the average fidelity, because this one is easier to work with, and we want to draw a parallel with the accuracy parameter in the diamond norm for the original UPQP. However, we believe that the symmetric structure of this task analogous to PBT will allow us to show that the optimal accuracy in terms of diamond norm will be equivalent to the one worst-case fidelity and average fidelity as in~\cite{yang_optimal_2020}.

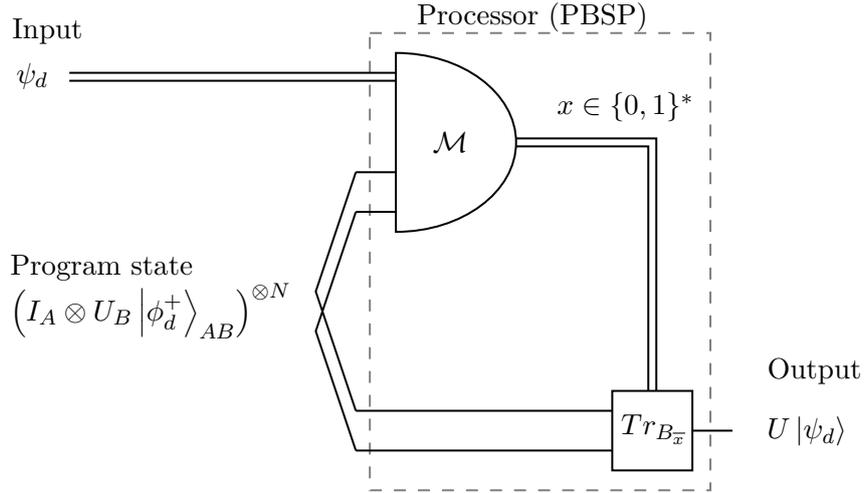
\begin{figure}
    \centering
    \tikzset{every picture/.style={line width=0.75pt}} 

\begin{tikzpicture}[x=0.75pt,y=0.75pt,yscale=-1,xscale=1]

\draw    (290,210) -- (270,150) -- (290,90) ;
\draw    (290,90) -- (310,90) ;
\draw    (290,210) -- (418,210) ;
\draw    (290,229.97) -- (270,169.97) -- (290,109.97) ;
\draw    (290,109.97) -- (310,109.97) ;
\draw    (290,230) -- (418,230) ;
\draw   (310,120) .. controls (310,120) and (310,120) .. (310,120) .. controls (343.14,120) and (370,99.85) .. (370,75) .. controls (370,50.15) and (343.14,30) .. (310,30) -- cycle ;
\draw    (147,40) -- (310,40) ;
\draw    (147,44) -- (310,44) ;
\draw   (370,73) -- (440,73) -- (440,200) ;
\draw   (370,77) -- (436,77) -- (436,200) ;
\draw   (418,200) -- (458,200) -- (458,240) -- (418,240) -- cycle ;
\draw    (458,220) -- (478,220) ;
\draw  [color={rgb, 255:red, 0; green, 0; blue, 0 }  ,draw opacity=0.5 ][dash pattern={on 4.5pt off 4.5pt}] (297,20) -- (467,20) -- (467,250) -- (297,250) -- cycle ;

\draw (327,68.4) node [anchor=north west][inner sep=0.75pt]    {$\mathcal{M}$};
\draw (119,33.4) node [anchor=north west][inner sep=0.75pt]    {$\psi _{d}$};
\draw (421,210.4) node [anchor=north west][inner sep=0.75pt]    {$Tr_{B_{\overline{x}}}$};
\draw (389,48.4) node [anchor=north west][inner sep=0.75pt]    {$x\in \{0,1\}^{*}$};
\draw (489,211.4) node [anchor=north west][inner sep=0.75pt]    {$\ U\ket{\psi _{d}}$};
\draw (117,11) node [anchor=north west][inner sep=0.75pt]   [align=left] {Input};
\draw (319,3) node [anchor=north west][inner sep=0.75pt]   [align=left] {Processor (PBSP)};
\draw (494,182) node [anchor=north west][inner sep=0.75pt]   [align=left] {Output};
\draw (116,130) node [anchor=north west][inner sep=0.75pt]   [align=left] {Program state \ \\ $\displaystyle \left( I_{A} \otimes U_{B}\ket{\phi _{d}^{+}}_{AB}\right)^{\otimes N}$$ $};

\end{tikzpicture}
    \caption{Construction of approximate UPHP from PBSP.}
\end{figure}

We now give an existence result by constructing a $\eps$-UPHP from PBSP. The UPHP constructed has a lower scaling of the memory dimension $m$ in terms of the error $\eps$ and the data dimension $d$, than the optimal scaling for $\eps$-UPQPs. We obtain this by lifting the exponential properties of PBSP, which once again shows that the former classical-quantum task is easier to perform by the processor than the fully quantum one, because we are not constrained by the no-programming theorem anymore.

\begin{theorem}\label{thm:construction_uphp}
	There exists a $d$-dimensional $\eps$-UPHP, $\P\eqdef\{\P_{\psi}\}_{\psi}\subset CPTP(\H_m,\H_ d\otimes \H_m)$, with memory dimension
	\begin{equation}
		{m\leq \left(\frac{1}{\eps}\right)^{4d\ln(d)}.}
	\end{equation}
\end{theorem}
\begin{proof}
 Given a unitary $U\in\U(\H_d)$, we can adapt the PBSP protocol from~\Cref{thm:fid_existence_pbsp} to transmit the state with the unitary applied to all ports instead, so that the $\eps$-UPHP will be a machine simulating the PBSP interaction between Alice and Bob, where Alice's classical PBSP input will be the one given to the processor, and the PBSP resource state shared between Alice and Bob will be used as the program state of the processor. Therefore, by using the Choi state of the unitary as a memory state, i.e.\ by using a PBSP protocol with $N$ copies of the state $(I_A\otimes U_B)\ket{\phi^+_d}_{AB}$ as a resource state, the measurements described by~\Cref{thm:fid_existence_pbsp} will output the unitary applied to the input state $U\dyad{\psi}U^\dag$ with high fidelity.
 	
By the properties of the maximally entangled state, we have that by projecting the unitary~$U$ applied to the maximally entangled state to the conjugate of the desired state $\ket{\psi^*}$, we obtain the unitary applied to the state with the same probability as projecting without the unitary. This holds because
	\begin{equation}\begin{split}
         \Tr_{A_i}\left[(\ketbra{\psi^*}_{A_i}\otimes I_{B_i})(I_{A_i}\otimes U_{B_i})\ketbra{\phi^+_d}_{A_iB_i}(I_{A_i}\otimes U_{B_i}^\dag)\right]&=\frac{1}{d}U\ketbra{\psi}_{B_i}U^\dag\\
         \Tr_{A_i}\left[((I-\ketbra{\psi^*}_{A_i})\otimes I_{B_i})(I_{A_i}\otimes U_{B_i})\ketbra{\phi^+_d}_{A_iB_i}(I_{A_i}\otimes U_{B_i}^\dag)\right]&=\frac{1}{d}U(I-\ketbra{\psi})_{B_i}U^\dag
	\end{split}\end{equation}
	Therefore, when Alice makes the measurements $\{M_x(\psi)\}_{x=1}^N$ as in~\Cref{eq:meas_dpbsp}, the probability of success of port $x\in[N]$ will be the same as the original PBSP~\Cref{eq:prob_meas}, but the outcome state will now have the unitary applied to it
	\begin{multline}
		\rho_{B_x}=\frac{1}{p_x}\left(\frac{1}{d}U\ketbra{\psi}_{B_x}U^\dag\sum_{i=0}^{N-1}\frac{\binom{N-1}{i}}{i+1}\left(\frac{1}{d}\right)^{i}\left(1-\frac{1}{d}\right)^{(N-1)-i}\right. \\
        \left.+\frac{1}{Nd}U(I-\ketbra{\psi})_{B_x}U^\dag\left(1-\frac{1}{d}\right)^{N-1}\right).
	\end{multline}
	Finally, the outcome fidelity between the protocol $\Phi_U(\psi)$ and the unitary applied to the system $U\ket{\psi}$ will be
     \begin{equation}\begin{split}
    		F\left(\Phi_U(\psi),U\ket{\psi}\right)&=\sum_{x=1}^Np_x\bra{\psi}_{B_x}U^\dag\rho_{B_x}U\ket{\psi}_{B_x}=\frac{N}{d}\sum_{i=0}^{N-1}\frac{\binom{N-1}{i}}{i+1}\left(\frac{1}{d}\right)^{i}\left(1-\frac{1}{d}\right)^{(N-1)-i}\\
    		&=1-\left(1-\frac{1}{d}\right)^N.
    	\end{split}
    \end{equation}

    We have thus constructed a processor $\P$ such that for every unitary $U\in\U(\H_D)$, there exists a memory state
    \begin{equation}
        \ket{\phi_U}_M\eqdef (I_A\otimes U_B){\bigotimes_{i=1}^N\ket{\phi^+_d}_{A_iB_i}},
    \end{equation}
    for which 
    \begin{equation}
        F\left(\Tr_M\left[\P_\psi({\ket{\phi_U}_M})\right], U\ket{\psi}_D\right) = F\left(\Phi_U(\psi), U\ket{\psi}\right)
        = 1-\left(1-\frac{1}{d}\right)^N,
    \end{equation}
    for all inputs $\psi$. The dimension of the memory $m$ of the $\eps$-UPHP is the {dimension of the maximally mixed state $d^2$} to the power of the number of ports $N$ needed $m\eqdef d^{2N}$. Therefore, if the number of ports is chosen such that {$N\geq d\ln(1/\eps^2)$}, from the properties of the exponential function and the Fuchs--van de Graaf inequality~\Cref{eq:fvdg}, we obtain
    \begin{equation}
		{\frac{1}{2}\norm{\Phi_U(\psi)-U\ketbra{\psi}U^\dag}_1%
        \leq\sqrt{1-F\left(\Phi_U(\psi), U\ket{\psi}\right)}%
        \leq \left(1-\frac{1}{d}\right)^{d\frac{\ln(1/\eps^2)}{2}}%
        \leq e^{\frac{\ln(\eps^2)}{2}}=\eps}
	\end{equation}
	which means that in order to obtain a processor with accuracy $\eps$, it is enough for the memory to be of dimension 
	\begin{equation}
		{m \leq \left(\frac{1}{\eps}\right)^{4d\ln(d)}.}
	\end{equation}
	
\end{proof}

\section{Achievable optimality of the tasks}\label{sec:optimality}
The existence results of the previous chapters show that having a classical description of a quantum state substantially increases the power of teleportation and processing tasks.
In this section, we complement these with bounds on the optimality achievable by Port-Based State Preparation and Universal Programmable Hybrid Processors.

\subsection{Approximate UPHP}

As we saw previously, an approximate UPHP is a quantum machine that, given a description of a quantum state, outputs a quantum state to which a unitary of choice is applied. In a way, $\eps$-UPHPs are encoding the functionality of the unitary in a quantum memory state, but this should intuitively also involve encoding at least some classical information into this quantum state.

We formalize this intuition in the following theorem, where we show the hardness of $\eps$-UPHP though the hardness of Quantum Random Access Codes.
That is, we show that UPHPs can be used as QRACs, which are limited by Nayak's bound,~\Cref{thm:nayak}.

\begin{theorem}\label{thm:optimal_uphp}
    For any $d$-dimensional $\eps$-UPHP, $\P\eqdef\{\P_{\psi}\}_{\psi}\subset CPTP(\H_m,\H_ d\otimes \H_m)$, the memory dimension must be at least
    \begin{equation}
        m\geq 2^{\frac{d}{2}(1-h(2\eps))}.
    \end{equation}
\end{theorem}
\begin{proof}
    Given a $(d,m,\eps)$-UPHP we can construct a $(d/2,\log(m),1-2\eps)$-QRAC. Let us consider an arbitrary $d/2$ bit-string $f\in\{0,1\}^{d/2}$ where each bit $f(x)$ is indexed by the bit-string $x\in\{0,1\}^{\log(d)-1}$. we can equivalently see this as a description of a Boolean function $f:\{0,1\}^{\log(d)-1}\to\{0,1\}$.

    We can compress this Boolean function into a $\log(m)$-dimension Hilbert space using the $(d,m,\eps)$-UPHP. That is, note that every Boolean function $f$ can be computed by a unitary
    \begin{equation}
        U_f\colon \ket{x}_1\otimes\ket{y}_2\mapsto \ket{x}_1\otimes\ket{y\oplus f(x)}_2,\quad\text{with }y\in\{0,1\},
    \end{equation}
    where the unitary associated with the Boolean functions of input size $\log(d)-1$ can be hosted in dimension $d$.

    Now from the $\eps$-UPHP assumption in~\Cref{eq:def_qcupp}, given $U_f\in\U(\H^{d/2})$ there exists a memory state $\ket{\phi_f}\in\H_M$ such that the processor $\P\eqdef\{\P_\psi\}_\psi\subset CPTP(\H_M,\H_M\otimes\H_D)$ approximates the desired unitary
    \begin{equation}
        {\frac{1}{2}\norm{\Tr_M\left[\P_x(\ketbra{\phi_f}_M)\right]-U_f(\ketbra{x}\otimes\ketbra{0})U_f^\dag}_1\leq\eps.}
    \end{equation}

    Now consider the channel consisting of applying $\P_x$, tracing out the parts we do not need, and then applying a measurement in the standard basis. Via a standard argument, we can write this procedure as a single measurement which we call $M_x^0$ and $M_x^1$. The exact form can be derived e.g.\ through writing out the Kraus operators.
    In particular, if define ${K^x_k}$ to be the Kraus operators of the channel $\P$ and subsequent partial trace, we have that $\Tr_{1,M}[\P_x(\cdot)] = \sum_k K^x_k \rho (K^x_k)^\dagger$. Now define the set of $d/2$ measurements $\{M_x^0,M_x^1\}_{x\in\{0,1\}^{\log(d)-1}}$:
    \begin{equation}
        M_x^y \eqdef \sum_k (K^x_k)^\dagger \ketbra{y} K^x_k, \quad\text{for }y\in\{0,1\}.
    \end{equation}
    We can verify that for these measurements it holds that
    \begin{equation}
        \Tr[M_x^y \rho] = \Tr[\sum_k (K^x_k)^\dagger \ketbra{y} K^x_k \rho] = 
        \Tr[\ketbra{y} \sum_k  K^x_k \rho (K^x_k)^\dagger   ] = 
        \Tr[\ketbra{y} \Tr_{1,M} [ \P_x(\rho)]],
    \end{equation}
    as desired.

    We can use these as the set of measurements of a $(d/2,\log(m),1-2\eps)$-QRAC, because the probability of correctly guessing the bit $f(x)$ of the bit-string $f$ is at least
    \begin{equation}\begin{split}
        \Tr[M_x^{f(x)}\ketbra{\phi_f}] &={ \Tr[\ketbra{f(x)}-\ketbra{f(x)}]+\Tr[M_x^{f(x)}\ketbra{\phi_f}]}\\
        &{=1-\Tr[\ketbra{f(x)}-\Tr_1[(I_1\otimes \ketbra{f(x)}_2) \Tr_{M} [ \P_x(\ketbra{\phi_f})]]]}\\
        &\geq1-\norm{\ketbra{f(x)}-\Tr_1[(I_1\otimes \ketbra{f(x)}_2) \Tr_{M} [ \P_x(\ketbra{\phi_f})]]}_1\\
        &\geq1-2\eps,
    \end{split}
    \end{equation}
    where in the last line we used the data processing inequality for the trace distance.

    Finally, from Nayak's bound,~\Cref{thm:nayak}, we know that, even when allowing some error, we cannot compress bits into less qubits, which leads to the desired bound
    \begin{equation}
        \log(m)\geq \frac{d}{2}(1-h(2\eps)),
    \end{equation}
    where $h$ refers to the binary entropy function.
    \end{proof}

    We remark that this bound (which of course also applies to $\eps$-UPQP, besides $\eps$-UPHP) is tighter than the bound in~\cite{kubicki_resource_2019}, that was proven in the context of $\eps$-UPQPs. We also want to point out that we were not able to reproduce the impossibility results blowing-up with the error tending to zero present in the UPQP literature~\cite{hillery_approximate_2006,majenz_entropy_2018,perez-garcia_optimality_2006}. These results seem to rely on the spectral decomposition of the processor, but the fact that we now have infinite copies of the input state makes it impossible to exploit the uniqueness of the spectral decomposition. In fact, it is not even clear to us how to prove the no-programming theorem as in~\cite{nielsen_programmable_1997} for the zero-error UPHP case.

\subsection{PBSP}

In the following, we give two bounds on the limitations of Port-Based State Preparation. The first bound follows from the achievable optimality of Universal Programmable Hybrid Processors, recall that the construction of PBSP from~\Cref{thm:construction_uphp} relied on UPHPs. The second bound follows from a non-signaling argument.

\begin{theorem}\label{thm:optimal_fidelity_pbsp}
	The optimal fidelity achievable by a {$(d,N,F^*)$-Deterministic} Port-Based State Preparation protocol with $N$ independent resource states is upper bounded by
	\begin{equation}
		{1-h(\sqrt{1-F^*})}\leq \frac{4N\log(d)}{d}.
	\end{equation}
\end{theorem}
\begin{proof}
Recall that we saw in~\Cref{thm:construction_uphp} how to use a {$(d,N,F=1-\eps^2)$}-PBSP to construct a $d$-dimensional $\eps$-UPHP with memory dimension $m=d^{2N}$.
Moreover, by~\Cref{thm:optimal_uphp} the dimension of the program register of $d$-dimensional $\eps$-UPHP is lower bounded by
\begin{equation}
    d^{2N}\geq 2^{\frac{d}{2}(1-h(2\eps))}.
\end{equation}
Substituting $\eps=(\sqrt{1-F})/2$, we obtain the desired bound
\begin{equation}
    {1-h\left(\sqrt{1-F}\right)\leq \frac{4N\log(d)}{d}.}
\end{equation}
\end{proof}

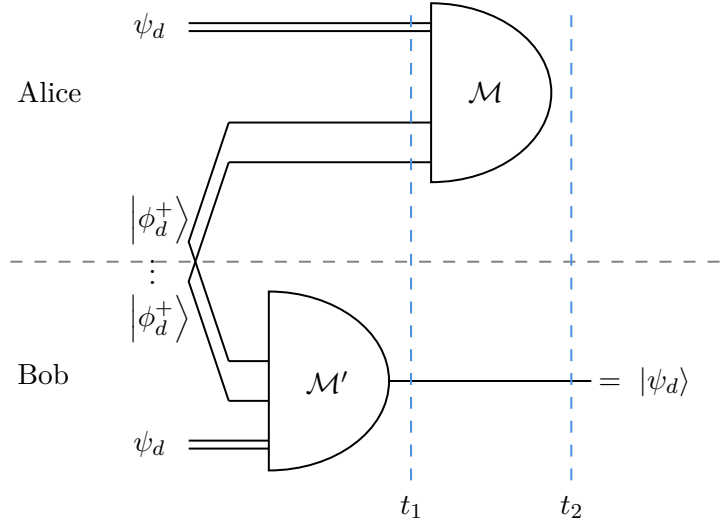
\begin{figure}
    \centering
    \tikzset{every picture/.style={line width=0.75pt}} 

\begin{tikzpicture}[x=0.75pt,y=0.75pt,yscale=-1,xscale=1]

\draw    (109,180) -- (89,120) -- (109,60) ;
\draw    (109,60) -- (210,60) ;
\draw    (109,180) -- (129,180) ;
\draw    (109,199.97) -- (89,139.97) -- (109,79.97) ;
\draw    (109,79.97) -- (210,79.97) ;
\draw    (109,200) -- (129,200) ;
\draw   (210,90) .. controls (210,90) and (210,90) .. (210,90) .. controls (243.14,90) and (270,69.85) .. (270,45) .. controls (270,20.15) and (243.14,0) .. (210,0) -- cycle ;
\draw    (89,10) -- (210,10) ;
\draw    (89,14) -- (210,14) ;
\draw    (189,190) -- (290,190) ;
\draw [color={rgb, 255:red, 0; green, 0; blue, 0 }  ,draw opacity=0.5 ] [dash pattern={on 4.5pt off 4.5pt}]  (0,130) -- (360,130) ;
\draw [color={rgb, 255:red, 74; green, 144; blue, 226 }  ,draw opacity=1 ] [dash pattern={on 4.5pt off 4.5pt}]  (200,240) -- (200,0) ;
\draw   (129,235) .. controls (129,235) and (129,235) .. (129,235) .. controls (162.14,235) and (189,214.85) .. (189,190) .. controls (189,165.15) and (162.14,145) .. (129,145) -- cycle ;
\draw [color={rgb, 255:red, 74; green, 144; blue, 226 }  ,draw opacity=1 ] [dash pattern={on 4.5pt off 4.5pt}]  (280,240) -- (280,0) ;
\draw    (89,220) -- (129,220) ;
\draw    (89,224) -- (129,224) ;

\draw (51,95.4) node [anchor=north west][inner sep=0.75pt]    {$ \begin{array}{l}
\ket{\phi _{d}^{+}}\\
\ \ \vdots \\
\ket{\phi _{d}^{+}}
\end{array}$};
\draw (227,38.4) node [anchor=north west][inner sep=0.75pt]    {$\mathcal{M}$};
\draw (60,3.4) node [anchor=north west][inner sep=0.75pt]    {$\psi _{d}$};
\draw (292,181.4) node [anchor=north west][inner sep=0.75pt]    {$=\ \ket{\psi _{d}}$};
\draw (2,38) node [anchor=north west][inner sep=0.75pt]   [align=left] {Alice};
\draw (2,180) node [anchor=north west][inner sep=0.75pt]   [align=left] {Bob};
\draw (146,183.4) node [anchor=north west][inner sep=0.75pt]    {$\mathcal{M} '$};
\draw (193,246.4) node [anchor=north west][inner sep=0.75pt]    {$t_{1}$};
\draw (273,246.4) node [anchor=north west][inner sep=0.75pt]    {$t_{2}$};
\draw (60,213.4) node [anchor=north west][inner sep=0.75pt]    {$\psi _{d}$};

\end{tikzpicture}
    \caption{Non-signalling times $t_1$ before Alice's measurement and $t_2$ after Alice's measurement.}\label{fig:non-signalling}
\end{figure}

For the non-signaling argument, the idea is for Bob to make a measurement on his half of the shared resources, and to compare the outcome before and after Alice makes a measurement, without any communication. Note that the following bounds explicitly use the fact that the resource states are maximally entangled states.

\begin{theorem}\label{thm:optimal_prob_pbsp}
    The optimal probability achievable by a $(d,N,p^*)$-Probabilistic Port-Based State Preparation protocol with $N$ maximally entangled resource states is upper bounded by
	\begin{equation}
		{p_\mathrm{EPR}^*}\leq 1-\left(1-\frac{1}{d}\right)^N.
	\end{equation}
\end{theorem}
\begin{proof}
    Imagine that Alice wishes to execute a PBSP protocol to transport {an arbitrary} $d$-dimensional state $\ket{\psi}$, that for the sake of the argument is also known to Bob.

    Imagine that Bob measures all his halfs of the $N$ maximally mixed states to see if they are in the state $\ketbra{\psi}$ or $I-\ketbra{\psi}$, for example by using the POVM $\mathcal{M}'=\{M_\top,M_\bot\}$, with
    \begin{equation}\label{eq:bob_bot}
        M_\bot(\psi) = {\bigotimes_{i=1}^N}(I-\ketbra{\psi})_{B_i}\quad\text{and}\quad M_\top = I-M_\bot.
    \end{equation}
    
    We will compare Bob's outcome before and after Alice does any measurement. Before Alice does any measurement, see time $t_1$ in~\Cref{fig:non-signalling}, observe because $\ket{\psi}$ is random, any specific port will have probability $\frac{1}{d}$ to see $\ket{\psi}$. Therefore, because the resource state is a product state of identical ports, the probability of seeing $\ket{\psi}$ in \emph{some} port is
    \begin{equation}\begin{split}
        p_\mathrm{success}^{t_1} &= 1-\Tr[(I_A\otimes M_\bot){\bigotimes_{i=1}^N}\ketbra{\phi_d^+}_{A_iB_i}] \\
        & = 1-{\prod_{i=1}^N}\Tr[(I_{A_i}\otimes (I-\ketbra{\psi})_{B_i})\ketbra{\phi_d^+}_{A_iB_i}]\\
        & = 1-\left(1-\frac{1}{d}\right)^N.
    \end{split}
    \end{equation}

    Let $\mathcal{M}=\{M_0,M_1,\ldots,M_N\}$ be Alice's POVM, outcome $x\in[N]$ refers to the state being in port $B_x$, and outcome $x=0$ refers to the protocol failing. After Alice does her measurement, see time $t_2$ in~\Cref{fig:non-signalling}, if the probabilistic PBSP went correctly, {a port $B_x$ would exist containing $\ket{\psi}$}. Therefore, the probability that Bob sees any port with the teleported state has to be at least $p^*_\mathrm{EPR}$ by correctness of the protocol. Formally, the probability of Bob seeing $\ket{\psi}$ in \emph{some} port is
    \begin{equation}\begin{split}
        p_\mathrm{success}^{t_2} & = \Tr[\left(I_A\otimes M_\top\right)\bigotimes_{i=1}^N\ketbra{\phi_d^+}_{A_iB_i}]\\
        & = \Tr[\left(M_0\otimes M_\top\right)\bigotimes_{i=1}^N\ketbra{\phi_d^+}_{A_iB_i}] + \Tr[\left(\sum_{x=1}^N M_x\otimes M_\top\right)\bigotimes_{i=1}^N\ketbra{\phi_d^+}_{A_iB_i}]\\ 
        & = \Tr[\left(M_0\otimes M_\top\right)\bigotimes_{i=1}^N\ketbra{\phi_d^+}_{A_iB_i}]+p^*_\mathrm{EPR}.
    \end{split}\end{equation}
    
    By non-signalling, the probabilities of some event happening on Bob's side (if no communication is involved) is identical before and after Alice's measurement. So by equating the probabilities of the event at time $t_1$ and $t_2$ we obtain the desired bound
    \begin{equation}
        1-\left(1-\frac{1}{d}\right)^N =p_\mathrm{success}^{t_1} = p_\mathrm{success}^{t_2}\geq p^*_\mathrm{EPR}.
    \end{equation}
\end{proof}

The same argument can be applied for deterministic PBSP, where the analogous non-signalling theorem for fidelity tells us that the entanglement fidelity of a protocol cannot increase if no communication happens between the parties. We will only include a sketch of the proof.

\begin{theorem}\label{thm:optimal_fidelity_pbsp_2}
    The optimal fidelity achievable by a $(d,N,F^*)$-Deterministic Port-Based State Preparation protocol with $N$ maximally entangled resource states is upper bounded by
	\begin{equation}
		{F_\mathrm{EPR}^*}\leq 1-\left(1-\frac{1}{d}\right)^N.
	\end{equation}
\end{theorem}
\begin{proof}
    We will compare Bob's outcome when performing the measurement~\Cref{eq:bob_bot} before and after Alice does the optimal measurement. Before Alice performs any measurement, see time $t_1$ in~\Cref{fig:non-signalling}, the probability of seeing $\ket{\psi}$ in some port is
    \begin{equation}
        p_{\mathrm{success}}^{t_1} = 1-\left(1-\frac{1}{d}\right)^N. 
    \end{equation}
    We can bound the probability of projecting into a tensor product $\bigotimes_{i=1}^N(I-\ketbra{\psi})_{B_i}$ by the minimum of all the marginal projections, which gives us an upper bound for the probability of failure in terms of individual registers
    \begin{equation}\begin{split}
        p^{t_1}_{\mathrm{failure}} &= 1- p^{t_1}_\mathrm{success}\\
        & = \Tr[\left(I_A\otimes\bigotimes_{i=1}^N(I-\ketbra{\psi})_{B_i}\right)\bigotimes_{i=1}^N\ketbra{\phi_d^+}_{A_iB_i}]\\
        & \leq \min_{\substack{x\in[N]}} \Tr[\left(I_{AB_{\bar{x}}}\otimes\left(I-\ketbra{\psi}\right)_{B_x}\right)\bigotimes_{i=1}^N\ketbra{\phi_d^+}_{A_iB_i}]\\
        & = \min_{\substack{x\in[N]}}(I-\ketbra{\psi})_{B_x}\Tr_{AB_{\bar{x}}}\left[\bigotimes_{i=1}^N\ketbra{\phi_d^+}_{A_iB_i}\right]
    \end{split}\end{equation}
    Let $\mathcal{M}=\{M_1,\ldots,M_N\}$ be Alice's POVM, outcome $x\in[N]$ refers to the state being in port $B_x$. By non-signalling a measurement on Alice's side without communicating the outcomes to Bob cannot change his perspective of the state, this is, for every $x\in[N]$:
    \begin{equation}
        \sum_{i\in[N]}\Tr_{AB_{\bar{x}}}\left[(M_i(\psi)\otimes I_B)\bigotimes_{i=1}^N\ketbra{\phi_d^+}_{A_iB_i}\right]= \Tr_{AB_{\bar{x}}}\left[\bigotimes_{i=1}^N\ketbra{\phi_d^+}_{A_iB_i}\right].
    \end{equation}
    Therefore we can rewrite the upper bound on the probability of failure as
    \begin{equation}\begin{split}
        p^{t_1}_{\mathrm{failure}}&=  \min_{\substack{x\in[N]}} \sum_{i\in[N]}(I-\ketbra{\psi})_{B_x}\Tr_{AB_{\bar{x}}}\left[(M_i(\psi)\otimes I_B)\bigotimes_{i=1}^N\ketbra{\phi_d^+}_{A_iB_i}\right]\\
        &\leq \sum_{x\in[N]}(I-\ketbra{\psi})_{B_x}\Tr_{AB_{\bar{x}}}\left[(M_x(\psi)\otimes I_B)\bigotimes_{i=1}^N\ketbra{\phi_d^+}_{A_iB_i}\right].
    \end{split}\end{equation}
    Note that by the correctness of the protocol, this value is related to the fidelity in the desired way
    \begin{equation}\begin{split}
        1 - F^*_\mathrm{EPR} & = 1- \sum_{x\in[N]}\bra{\psi}_{B_x}\Tr_{AB_{\bar{x}}}\left[(M_x(\psi)\otimes I_B)\bigotimes_{i=1}^N\ketbra{\phi_d^+}_{A_iB_i}\right]\ket{\psi}_{B_x}\\
        & = \sum_{x\in[N]} p_x - \Tr[\ketbra{\psi}_{B_x}\Tr_{AB_{\bar{x}}}\left[(M_x(\psi)\otimes I_B)\bigotimes_{i=1}^N\ketbra{\phi_d^+}_{A_iB_i}\right]]\\
        & =\sum_{x\in[N]} \Tr[(I-\ketbra{\psi})_{B_x}\Tr_{AB_{\bar{x}}}\left[\left(M_x(\psi)\otimes I_B\right)\bigotimes_{i=1}^N\ketbra{\phi_d^+}_{A_iB_i}\right]].
    \end{split}\end{equation}
\end{proof}

\section*{Acknowledgement}

We thank Marco Túlio Quintino for his help navigating the PBT literature and introducing us to high-order physics, and Jonas Helsen for valuable discussions. FS was supported by the Dutch Ministry of Economic Affairs and Climate Policy (EZK), as part of the Quantum Delta NL programme.

\printbibliography

\end{document}